\documentclass[pra,aps,amsmath,amssymb,superscriptaddress,twocolumn,longbibliography]{revtex4-1}
\usepackage{graphicx,multirow}
\usepackage{color,outlines}
\usepackage[english]{babel}
\usepackage{amsthm}
\usepackage[sc,osf]{mathpazo}
\usepackage{hyperref}
\usepackage{physics}
\usepackage{outlines}
\usepackage{tikz,tikz-qtree}

\newcommand{\bla}{\color{black}}

\newcommand{\e}{\mathcal{E}}

\newtheorem{thm}{Theorem}
\newtheorem{corl}{Corollary}
\newtheorem{definition}{Definition}

\definecolor{green2}{RGB}{0,100,0}
\newcommand{\ml}{\mathcal{L}} 
\newcommand{\md}{\mathcal{D}} 

\begin{document}

\title{Ergodic and mixing quantum channels: From two-qubit to many-body quantum systems}

\author{S. Aravinda } 
\email{aravinda@iittp.ac.in}
\affiliation{Department of Physics, Indian Institute of Technology Tirupati, Tirupati, India~517619} 

\author{Shilpak Banerjee} 
\email{shilpak@iittp.ac.in}
\affiliation{Department of Mathematics and Statistics, Indian Institute of Technology Tirupati, Tirupati, India~517619} 

\author{Ranjan Modak} 
\email{ranjan@iittp.ac.in} 
\affiliation{Department of Physics, Indian Institute of Technology Tirupati, Tirupati, India~517619} 

\begin{abstract}
The development of classical ergodic theory has had a significant impact in the areas of mathematics, physics, and, in general, applied sciences. The quantum ergodic theory of Hamiltonian dynamics has its motivations to understand thermodynamics and statistical mechanics.  Quantum channel, a completely positive trace-preserving map, represents a most general representation of quantum dynamics and is an essential aspect of quantum information theory and quantum computation.  In this work, we study the ergodic theory of quantum channels by characterizing different levels of ergodic hierarchy from integrable to mixing. The quantum channels on single systems are constructed from the unitary operators acting on bipartite states and tracing out the environment. The interaction strength of these unitary operators measured in terms of operator entanglement provides sufficient conditions for the channel to be mixing. By using block diagonal unitary operators, we construct a set of non-ergodic channels. 
%From integrable to mixing is characterized explicitly in the case of the two-qubit unitary operator. 
By using canonical form of two-qubit unitary operator, we analytically construct the channels on single qubit ranging from integrable to mixing. 
Moreover, we also study interacting many-body quantum systems that include the famous Sachdev-Ye-Kitaev (SYK) model and show that they display mixing within the framework of the quantum channel. 

\end{abstract}

\maketitle

\section{Introduction} 

 In classical ergodic theory, we study the statistical properties of a probability-measure-preserving transformation, and the five most widely studied properties in the community are ergodicity, weak-mixing, (strong) mixing, K-property, and Bernoulli. It turns out that these properties can be stacked in a nice hierarchy, often referred to as the ergodic hierarchy, with the strongest one being Bernoulli. More precisely: Bernoulli$\implies$K-property$\implies$(strong) mixing$\implies$weak mixing$\implies$ergodicity~\cite{cornfeld1982ergodictheory}. The reverse implications are not true in general and counterexamples exist.  Predictable, the set of transformations becomes smaller when one requires a stronger property in the hierarchy, and in fact, invertible Bernoulli transformations are completely classified~\cite{ornsetein1970bernoulli} while invertible ergodic transformations cannot be classified~\cite{foreman2011conjugacy} answering the isomorphism question laid down in von Neumann's seminal paper~\cite{vonNeumann1932Zur} in the negative.

In quantum theory, early works on quantum ergodic theory in Hamiltonian unitary dynamics are by von Neumann which he named it as quantum ergodic theorem and H theorem~\cite{von2010proof,neumann1929beweis}. These theorems met with criticisms~\cite{zaslavsky1981stochasticity,PeresI84,PeresII84} and counter justifications~\cite{goldstein2010long}. Peres defined and studied ergodicity and mixing based on his definition of quantum chaos~\cite{PeresI84,PeresII84}. Many others continued to study ergodic and mixing behaviour in quantum Hamiltonian dynamics, maps, and its relation to chaos and various other related quantities~\cite{prosen1998time,artuso1997numerical,casati1999mixing,prosen1999ergodic,prosen2002general,horvat2009ergodic,zhang2016ergodicity}.   In the quantum many-body systems, these concepts link with that of the thermalization and has produced a host of interesting results (see review~\cite{qc16rev} and reference therein). Remarkably, certain spacetime dual quantum many-body circuit provides an analytically solvable minimal model for the correlation function between two spatiotemporal observables~ \cite{Akila2016,Bertini2019,gutkin2020exact,BraunPRE2020}, and thereby facilitates the characterization of the entire ergodic hierarchy~\cite{Bertini2019,ASA_2021}. The question of how an isolated many-body quantum system thermalizes while going through unitary dynamics also brought lots of interest in the last two decades thanks to extraordinary advancements in cold atom experiments~\cite{RevModPhys.80.885,RevModPhys.83.1405}. 
It is now well-established that
if the many-body energy level spacing of a system are of the Wigner-Dyson
type, one expects thermalization to occur~\cite{rigol.10} i.e. the system acts as its own bath, obey the eigenstate thermalization Hypothesis~\cite{rigol.07,rigol.12,sredniki.94,deutsch.91}. 
%One of the goals of this work is to address the questions of thermalization for many-body systems within the framework of quantum channels.   

Quantum channels, a linear map that maps quantum states, is a generalization of unitary quantum dynamics~\cite{wolf2012quantum,wilde2013quantum}. Mathematically, it is a linear, completely positive trace preserving (CPTP) map~\cite{paulsen2002completely,bhatia2009positive}. At this juncture, we would like to map the results of ergodic theory from classical to quantum. Here, the quantum channel acts as a transformation on the set of quantum states, and we are interested in studying the iterative application of quantum channels. 
Ergodic theory of quantum channels has been studied in the past~\cite{alicki2007quantum,holevo2003statistical} using the algebraic framework. With the emergence of quantum information theory and computation, the ergodic theory of quantum channels in finite-dimensional systems took the forefront as it has been the topic of many recent works~\cite{burgarth2007generalized,burgarth2013ergodic,movassagh2020ergodic,movassagh2021theory,singh2022ergodic}. The theoretical results of iterative application of quantum channels has been used in various tasks in quantum information theory. For example, in the studies of structure-preserving maps, which have led to a better understanding of quantum error correction~\cite{raginsky2002strictly,blume2008characterizing,blume2010information}, the convergence rates of quantum channels are related to critical exponents in the tensor representation of many-body systems~\cite{giovannetti2008quantum},  to construct any measurement using sequential generalized measurements~\cite{ma2023sequential,linden2022use}, and asymptotic results in quantum channels~\cite{albert2019asymptotics,amato2023asymptotics}.  Mixing of quantum channels is related to information-theoretic channel capacities~\cite{singh2024zero}.

\begin{figure}
    \centering
    \includegraphics[scale=0.52]{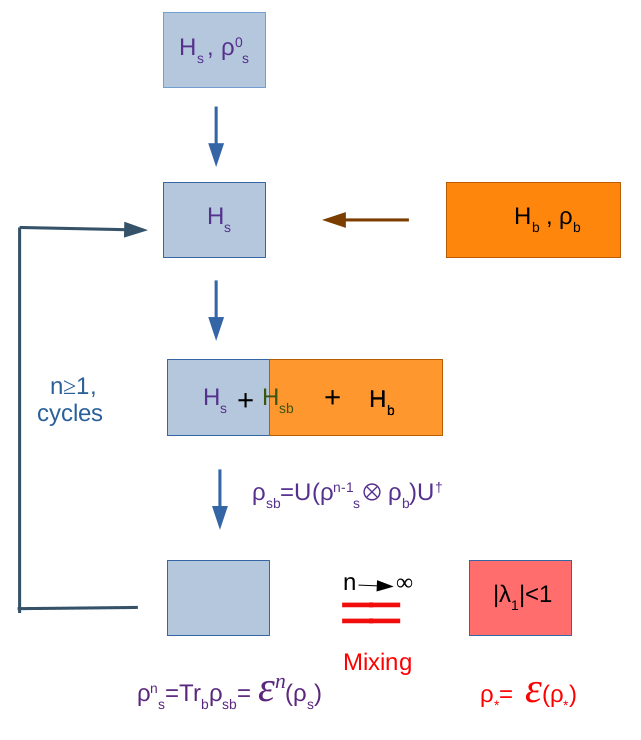}
    \caption{ Sketch of the iterative process studied in this work (see text Sec.~\ref{summary}
for the description).}
    \label{schematic}
\end{figure}
\section{Summary of main results\label{summary}}

In this section, we first briefly describe our setup and then summarize our main findings. We prepare a system that is described by the Hamiltonian $H_s$ and state $\rho_s$. Then, we connect the system with a bath and let the entire system (system+bath) evolve. $\rho_b$ and $H_b$ are referred to as the state and the Hamiltonian of the bath, respectively, and $H_{sb}$ stands for the interaction between the system and the bath. After the evolution, the state of the entire system becomes, 
$\rho_{sb}=U (\rho \otimes \sigma ) U^\dagger$, where $U=e^{-i(H_s+H_b+H_{sb})}$. We trace out the bath degrees of freedom from $\rho_{sb}$ and obtain the reduced density matrix for the system  $\rho^1_s=\Tr_b\rho_{sb}$.  Again, we connect the system with the bath and repeat it iteratively $n$ times. Schematically, this is explained in Fig.~\ref{schematic}. 

In a broader theoretical perspective, every quantum channel acting on the single quantum system has a representation in terms of a unitary operator acting on a bipartite system and then tracing out one of the systems~\cite {Zyczkowski2004,Bengtsson2007}. %The representation can be treated as a unitary operator acting jointly on the product state of the system and environment, and by tracing out the environment's degrees of freedom, the action of the channel on the state is realized. Conversely, every unitary operator acting on the bipartite system can be treated to induce a quantum channel in the single-partite system.  
By using the notions of matrix realignment~\cite{oxenrider1985matrix,sudarshan1961stochastic,Zyczkowski2004} used extensively in entanglement theory~\cite{chen2002matrix,rudolph2003some} and recently in the many-body systems literature~\cite{SAA2020,RAA22}, we provide the representation of the channel in terms of the unitary operator and the environment, and is explained in Sec.~(\ref{sec:chanrep}). Our study focuses largely on the spectral properties of channels, and the details are presented in Ref.~\cite{wolf2012quantum,wilde2013quantum}. For completeness and continuity, we present the spectral results of the channel in detail in Sec.~(\ref{sec:chanspec}), and we define the various equivalent definitions of ergodic and mixing in Sec.~(\ref{sec:ergdef}). The main results of our paper at the formalism level and application to many-body physics are presented in Sec.~(\ref{sec:erg_res}) and Sec.~(\ref{sec:manybody}), respectively. The sufficiency condition on mixing of quantum channels in terms of the strength of unitary operator measured as operator entanglement~\cite{Zanardi2001,Zanardi2000} is proved in Sec.~(\ref{sec:suff_erg}), the stability of sufficiency condition under any local unitary operators is presented in Sec.~(\ref{sec:stab}), ergodic and non-ergodic channels are constructed by using block diagonal unitary operators in Sec. (\ref{sec:diagon}), and bounds on the convergence of any initial state to the fixed state in Sec.~(\ref{sec:conv}). Finally, we construct various single qubit channels using two-qubit unitary operators as an illustration in Sec.~(\ref{sec:qubits}) and
 also investigate many-body quantum Hamiltonians (e.g. the Sachdev-Ye-Kitaev (SYK) model~\cite{sachdev.2015,sachdev.93} and experimentally realizable short-ranged model of interacting fermions in the presence of quasi-periodic potential~\cite{schreiber2015observation} in Sec.~(\ref{sec:manybody}).   Strikingly, all the systems studied here (which involve both short-ranged and long-ranged models) show mixing, and we could not find an example where the system is ergodic but not mixing.
  Moreover, we find the mixing speed decreases with the increase of the Hilbert space dimension, and in parallel, it also becomes harder for the system to meet the sufficiency condition of the mixing, which was derived in Sec.~(\ref{sec:suff_erg}) in terms of the operator entanglement. 
Finally, we also qualitatively predict the scrambling time for such systems using the generalized spectral form factor (SFF).

%In this setting, we define ergodic and mixing of quantum channels and provide the sufficient condition (see Eq.~(\ref{eu_bound}) for mixing in terms of the operator entanglement $E(U)$  of the unitary operator $U$. It implies, in the $n\to \infty$ limit, the reduced density matrix converges to the fixed state $\rho_*$ such that $\e (\rho_*)=\rho_*$, where $\e$ is a quantum channel whose second largest eigenvalue   $|\lambda_1|<1$. The convergence of any state to a fixed state is bounded by the operator entanglement of unitary operator (see Eq.~(\ref{eq:convbound}). We construct the non-ergodic channels using block-diagonal unitary operators.  
%We consider the two-qubit unitary operator $U=U_{nl}$ (see Eq.~\eqref{eq:cartan}) and as well as unitary operators obtained from the many-body Hamiltonians e.g. SYK (see Eq.~\eqref{hamiltonian2}) and SR (see Eq.~\eqref{hamiltonian1}) models.  We find that all these systems display mixing. For two-qubit unitary operator, we characterized all different phases of the ergodic hierarchy using its Cartan coefficients.

% We also derive a sufficient condition on the Unitary operator $U$ (see Eq.~\eqref{eu_bound}) that guarantees mixing.  

\section{Quantum Channels \label{formalism}}
In this Section, we provide a brief introduction to the quantum channels and their spectral properties. In Sec.~(\ref{sec:chanrep}), various representations of quantum channels and their interrelations are explained by explicitly expressing the linear operator acting on a single system in terms of the bipartite unitary operator. In Sec.~(\ref{sec:chanspec}), we reproduce the spectral results for the quantum channel, which is required for our work. In Sec.~(\ref{sec:ergdef}), we explain the definition of ergodic and mixing quantum channels. 

\subsection{Representation of quantum channels \label{sec:chanrep}}

Let  $\mathcal{H}^d$ is the Hilbert space of dimension $d$. The quantum states are represented by the density matrix  
$\rho$ such that $\rho \geq 0$ and $\Tr (\rho) = 1$, and the state space $\mathcal{S}(\mathcal{H})$ is defined over  the set of density matrices. The map $\e$ is a linear, completely positive trace preserving (CPTP),  defined over the state space $\mathcal{S}(\mathcal{H})$, and is called a quantum channel~\cite{wilde2013quantum}. The action of $\e$ on the state $\rho$ is 
\begin{equation}
\rho^\prime = \e (\rho), 
\label{eq:chn1}
\end{equation}
and its corresponding linear operator is represented as $\ml$~\cite{Bengtsson2007} 
\begin{equation}
    \ket{\rho^\prime} = \ml \ket{\rho}
    \label{eq:chn2}
\end{equation}
where the vectorization is defined as $\mel{i}{\rho}{j} = \braket{ij}{\rho}$. The linear operator $\ml$ is a $d^2 \times d^2$ matrix, and we use bipartite index notation to represent the matrix, $\mel{i\alpha}{\ml}{j\beta} \equiv \ml_{i\alpha,j\beta}$. Note that the operator $\ml$ is not Hermitian. The channel can be represented using an Hermitian operator,  {\it dynamical matrix}~\cite{sudarshan1961stochastic,Zyczkowski2004} $\md$, defined as 
\begin{equation}
\md = \ml^{R_2}, 
\end{equation} 
where $R_2$ is a matrix reshaping operation defined as $\mel{i\alpha}{A}{j\beta} = \mel{ij}{A^{R_2}}{\alpha \beta}$. Expressing $\md$ in its eigenbasis as 
\begin{equation}
    \md = \sum_i  \ketbra{A_i},
\end{equation}
where we rescaled the basis as $\ket{A_i} = \sqrt{g_i} \ket{\gamma_i}$, with $g_i$ and $\ket{\gamma_i}$ are eigenvalues and eigenstates of $\md$ respectively. The operator $\ml$ we can be  written as  
\begin{equation}
    \ml  = \sum_i g_i \gamma_i \otimes \gamma_i^* = \sum_i A_i \otimes A_i^* .
\end{equation}
Here we used the identity $(A\otimes B)^{R_2} = \ketbra{A}{B^*}$. This provides a canonical form of Kraus representation with $\sum_i A_i^\dagger A_i = I$, 
\begin{equation}
    \e(\rho) = \sum_i A_i \rho A_i^\dagger. 
    \label{eq:kraus}
\end{equation}
The operator $\md$ is also called as {\it Choi} matrix~\cite{choi1975completely} and represented as 
\begin{equation}
    \md = (\e \otimes I) \ketbra{I},
\end{equation}
with the action of the channel can be recovered as 
\begin{equation}
    \e (\rho) = \Tr_2 [(I \otimes \rho^T) \md],
\end{equation}
where $T$ is the matrix transposition operation,  $I$ is the identity operator, and $\ket{I} = \sum_i \ket{ii}$.  

Every channel $\e$ has its unitary representation in a larger dimension. The channel $\e$ acting on a single system in $\mathcal{H}^d$ is obtained from the unitary operator acting on the bipartite system in $\mathcal{H}^d \otimes \mathcal{H}^d$  by tracing the second system,  
\begin{equation}
 \e (\rho) = \Tr_2 [U (\rho \otimes \sigma ) U^\dagger].  
 \label{eq:uni_ex}
\end{equation}
This represents the physical scenario in which the joint system evolves under the unitary operator $U$, in which the second system can be treated as an environment (or bath). We are interested in the dynamics of the system alone.  %The complmentary channel can also be defined in which we are interested in the environment so by tracing the first system we get a channel which is called complementary channel (see App)  and it has several important application in the cryptographic scenario. 
Let $\sigma = \sum_{ij}  \lambda_{ij} \ketbra{i}{j}$ and for the channel $\e (\rho)$ in Eqn.~(\ref{eq:uni_ex}), the matrix representation is 
\begin{equation}
 \begin{split}
  & \bra{\alpha}  \e(\rho) \ket{\beta}  = \sum_\gamma \bra{\alpha \gamma} U (\rho \otimes \sigma ) U^\dagger \ket{\beta \gamma } \\
  &= \sum_{\gamma ij mk} \lambda_{ij} \bra{\alpha \gamma }U \ket{mi} \bra{kj} U^\dagger \ket{\beta \gamma} \bra{m}\rho \ket{k} \\
  &= \sum_{\gamma ij mk} \lambda_{ij} \bra{\alpha m } U^{R_2} \ket{\gamma i}  \bra{\gamma j} U^{R_2\dagger} \ket{\beta k} \braket{mk}{\rho} \\
  &= \sum_{ mk}  \bra{\alpha \beta } \left[ U^{R_2} (I \otimes \sigma ) U^{R_2 \dagger} \right]^{R_2} \ket{mk}\braket{mk}{\rho}    
 \end{split}
\end{equation}
By referring Eq.~(\ref{eq:chn1}) and Eq.~(\ref{eq:chn2}), the matrix representation of $\ml$ in terms of the unitary operator $U$ is 
\begin{equation}
 \ml = \left[ U^{R_2} (I \otimes \sigma ) U^{R_2 \dagger} \right]^{R_2}.
 \label{eq:L}
\end{equation}

\subsection{Spectral properties of channels \label{sec:chanspec}}
The ergodic properties of a channel depend largely on the spectral properties of quantum channels. For completeness, we comprehensively explain these results here~\cite{wolf2012quantum,wilde2013quantum}. 

The eigenvalue equation of the channel $\e$ with eigenvalue $\lambda$ and an eigenoperator  $X$ is 
\begin{equation}
    \e (X) = \lambda X. 
\end{equation}
The eigenvalues and eigenoperators are the solution of the characteristic equation $\det (\lambda I - \sum_i A_i \otimes A_i^*) =0 $. Denote $\texttt{spec} (\e)$ as the spectrum of the channel $\e$. However, note that the channel $\ml$ need not be diagonalizable. 

\begin{thm}
    The channel $\e$ has a density matrix $\rho_*$ as a  fixed point. 
\end{thm}
\begin{proof}
    The proof follows from the fact that the channel $\e$ is linear continuous mapping on the set of density matrices. Since the set of density matrices are convex and compact, due to Markov-Kakutani theorem [V.10.6 of Ref.~\cite{dunford1988linear} ], the theorem  follows~\cite{terhal2000problem}. 
\end{proof}

\begin{corl}
    If $\lambda_0$ is an eigenvalue of the channel with the density matrix as its eigenstate say $\rho_*$, i.e., $\e (\rho_*) = \lambda_0 \rho_*$, then $\lambda_0 = 1$
\end{corl}
\begin{proof}
    Since $\Tr{\rho_*}=1$, the proof follows from the trace-preserving property of the channel. 
\end{proof}

\begin{thm} \label{thm:comspec}
 $\lambda$ is an eigenvalue of the channel $\e$ with an eigenvector $X$ iff $\lambda^*$ is also an eigenvalue with eigenvector $X^\dagger$. 
\end{thm}
 \begin{proof} Using the Kraus representation, we can express the action of the channel as  $\mathcal{E}(\rho)=\sum_i A_i\rho A_i$ where $\sum_iA_i^\dagger A_i=I$. So, starting with the eigenvalue equation, 
    \begin{equation}
        \begin{split}
            &\mathcal{E}(X)=\lambda X\\
            \iff & \sum_i A_i X A_i^\dagger =\lambda X\\
            \iff & \left(\sum_i A_i X A_i^\dagger \right)^\dagger = \left(\lambda X\right)^\dagger \\
            \iff & \sum_i A_i X^\dagger A_i^\dagger = \lambda^* X^\dagger\\
            \iff & \mathcal{E}(X^\dagger)=\lambda^* X^\dagger
        \end{split}
    \end{equation}
\end{proof}
If $\e$ is a linear map, the adjoint map $\e^\dagger$ is defined as  $\e^\dagger (X) = \sum_i A_i^\dagger X A_i$. The adjoint map $\e^\dagger$ of a completely positive map is also completely positive~\cite{wilde2013quantum}. The adjoint of a CPTP map, a quantum channel, is an unital map. If the channel $\e$ is unital, then its adjoint $\e^\dagger$ is trace-preserving and thereby forms an unital channel.  The spectrum of $\e$ and $\e^\dagger$ is related as follows: $\texttt{spec}(\e^\dagger) = \texttt{spec} (\e)^*$, due to which a channel $\e$ and its adjoint $\e^\dagger$ has same spectrum.

\begin{thm}
    The eigenvalues of the channel $\e$ are bounded as: 
    \begin{equation}
        \abs{\lambda} \leq 1. 
    \end{equation}
\end{thm}
\begin{proof}
    First we switch from the Schr\"odinger picture for $\e$ to the Heisenberg picture for $\e^\dagger$ induced by 
    \begin{equation}
        \begin{split}
            \Tr[\e(X)Y]=\Tr[X\e^\dagger(Y)]
        \end{split}
    \end{equation}
    We also recall that linearity and complete positivity of $\e$ translates to linearity and complete positivity of $\e^\dagger$ whereas the trace-preserving property of $\e$ becomes unitality ($\e^\dagger(I)=I)$ of $\e^\dagger$, and moreover, both $\e$ and $\e^\dagger$ have the same spectrum. So, proving the statement for $\e^\dagger$ will suffice.
    
    Since for any $X\in\mathcal{H}^d$, $X/\norm{X}_\infty$ is a contraction, a consequence of the theorem of Russo and Dye \cite{RussoDye1966} can be stated as follows \cite{bourin2019russo}:
    \begin{equation}
        \norm{\e^\dagger(X)}_\infty \leq \norm{\e^\dagger (I)}_\infty \norm {X}_\infty . 
    \end{equation}
    Here $\norm{\cdot}_\infty$ is an operator norm and for matrix it is defined as the maximum absolute row sum of the matrix. In the particular case where $X$ is an eigenvector of $\e^\dagger$ with eigenvalue $\lambda$, we have 
    \begin{equation}
        \begin{split}
            & \abs{\lambda} \norm{X}_\infty  \leq \norm{\e^\dagger (I)} \norm{X}_\infty \\ 
            \implies & \abs{\lambda} \leq \norm{\e^\dagger (I)}_\infty = \norm{I}_\infty =1
        \end{split}
    \end{equation} 
\end{proof}

The ergodic properties of quantum channels largely depend on their spectrum. Arrange the spectrum of $\e$ in the increasing order of magnitude as
\begin{equation}
    \lambda_0 = 1 \geq \abs{\lambda_1} \geq \abs{\lambda_2} \cdots \abs{\lambda_{d^2-1}},
    \label{order}
\end{equation}
which will be used to define the various properties of quantum channels.

\subsection{Definitions: Ergodic and mixing quantum channels \label{sec:ergdef}}
 Ergodic theory is a mathematically well-characterized field~\cite{cornfeld1982ergodictheory}. All  definitions related to the ergodic hierarchy are explained in various sources~\cite{berkovitz2006ergodic}. Here, we define ergodic and mixing of quantum channels~\cite{singh2022ergodic,burgarth2013ergodic}. The essential aspect of the dynamical system is the measure space equipped with transformation. In the ergodic theory of quantum systems~\cite{alicki2001quantum},  the set of events is represented by the positive operator with the measure given by trace rule due to Gleason's theorem~\cite{gleason1975measures,busch2003quantum,barnett2014quantum}. We focus on quantum channels as the transformations. 

%Let us denote $\mathcal{M}(d)$ as a set of complex $d \times d$ complex matrices. 
The composition of the channel $\e_1$ and $\e_2$ is denoted as $(\e_1 \circ \e_2) (\rho) = \e_1(\e_2(\rho))$. Note that if $\e$ is CPTP map, then any iterative application of the map 
\begin{equation}
    \e^n := \e \circ \cdots \circ \e 
\end{equation}
is also a CPTP map~\cite{wilde2013quantum}. The  following quantity, similar to taking an average, is also a quantum channel, 
\begin{equation}
    \Lambda_N^\e (\rho) := \frac{1}{N+1} \Sigma_{n=0}^N \e^n (\rho),
\end{equation}
where $\e^0 $, an identity channel.

\begin{definition}[Ergodic~\cite{burgarth2013ergodic}] 
A quantum channel is ergodic if there exists a unique state $\rho_*$ such that 
\begin{equation}
\e (\rho_*) = \rho_* . 
\end{equation} 
\end{definition}
 Here $\rho_*$ is the unique fixed point. In terms of the spectrum, the channel $\e$ is ergodic if and only if it has $\lambda_0 =1$ and $\lambda_i \neq 1, \forall i \neq 0$.   The ergodicity of the channel can be defined through the convergence of $\Lambda_N^\e$. The channel $\e$ is ergodic if and only if 
\begin{equation}
    \lim_{N\rightarrow \infty} \norm{\Lambda_N^\e (\rho) - \rho_*}_1 = 0, 
\end{equation}
where $\norm{A}_1 := \Tr[\sqrt{A^\dagger A}]$ is the trace norm. It instructs us to see the convergence in terms of the function of operators. Let $\mathcal{M}_d$ is the set of $d \times d$ complex matrices.  For any $A,B \in \mathcal{M}_d $, the quantum channel $\e$ is ergodic if and only if the following equivalent conditions holds~\cite{singh2022ergodic}, 
\begin{equation}
    \begin{split}
        \lim_{N\rightarrow \infty}  \Lambda_N^\e(A) &= \Tr (A) \rho_*  \\
        \lim_{N\rightarrow \infty} \frac{1}{N+1} \Sigma_{n=0}^N \Tr (\e^n(A)B) &= \Tr(A)\Tr(B\rho_*). 
    \end{split}
\end{equation}

\begin{definition}[Mixing~\cite{burgarth2013ergodic}]
    $\e $ is mixing if there exists a unique state $\rho_*$ such that 
    \begin{equation}
        \lim_{n \to \infty} \norm{\e^n (\rho) - \rho_*}_1 = 0 \quad \forall \rho 
    \end{equation}
\end{definition}
In terms of the spectrum, the channel $\e$ is mixing if and only if its spectrum has the property that $\abs{\lambda_i} \neq 1, \forall i \neq 0$. Similar to the ergodic property, for any $A,B \in \mathcal{M}_d$, the equivalent definition of mixing holds if and only if the following equivalent conditions are satisfied, 
\begin{equation}
    \begin{split}
     \lim_{n\rightarrow \infty} \e^n(A) &= \Tr (A) \rho_* \\
        \lim_{n\rightarrow \infty} \Tr (\e^n(A)B) & = \Tr(A) \Tr(B\rho_*).       
    \end{split}
\end{equation}

Mixing implies ergodicity but not the converse. If the spectrum is such that there exist $\lambda_k$ for which $\lambda_k \neq 1$ but $\abs{\lambda_k} =1$. Such channels are ergodic but not mixing. In classical ergodic theory, further classifications of mixing channels like weak mixing and strong mixing are defined. It has been shown recently that in the case of quantum channels, weak mixing and strong mixing is equivalent to mixing~\cite{singh2022ergodic}, and we focus on mixing in this work. % {\blu The information theoretic consequences of mixing for calculating channel capcities are interesting and we explained this in the Appendix~(\ref{sec:app_info}).}

\section{Ergodicity, Mixing and operator entanglement \label{sec:erg_res}} 
In this Section, we study the ergodic and mixing nature of quantum channels defined in Eq.~\ref{eq:uni_ex}, by considering the spectrum of the quantum channel with its representation in terms of the unitary operator as in Eq.~\ref{eq:L}. 

\subsection{Sufficient condition for mixing \label{sec:suff_erg} }

The Hilbert-Schmidt norm or Frobenius norm of an operator $A$ is defined as $\norm{A} = \sqrt{\Tr(AA^\dagger)}$. By using an easily provable identity $\Tr(X^{R_2}X^{R_2\dagger}) = \Tr(XX^\dagger)$ for any square matrix $X$, we can see 
\begin{equation}
    \norm{\ml}^2 = \Tr\left[ (U^{R_2} (I \otimes \sigma ) U^{R_2\dagger})^2 \right]. 
    \label{eq:norm} 
\end{equation}
We are more interested in the case in which $\sigma = I/d$, and for this we have 
\begin{equation}
 \norm{\ml}^2 = \frac{1}{d^2}\Tr\left[ (U^{R_2}  U^{R_2\dagger})^2 \right].
\end{equation}

In order to understand the mixing of quantum channel $\ml$ as a function of the unitary operator $U$, it is required to understand the entangling nature of the unitary operator. Physically, it encodes the interaction that the bipartite unitary operator $U$ induces between the system and environment. Consider an operator Schmidt decomposition~\cite{nielsen2003quantum} of a unitary operator $U$
\begin{equation}
 U = \sum_i \sqrt{\mu_i} A_i \otimes B_i,
\end{equation}
where $\mu_i \geq 0$ and $\Tr (A_i A_j^\dagger) = \Tr (B_i B_j^\dagger) = \delta_{ij}$. The unitarity of $U$ implies that $\sum_{i=0}^{d^2-1} \mu_i = d^2$, and hence $\mu_i/d^2$ can be treated as probabilities and related entropies can be defined. By using linear entropy,  the normalized operator entanglement $E(U)$ is defined as~\cite{Zanardi2001}, 
\begin{equation}
 E(U) = \left(\frac{d^2}{d^2-1}\right) \left(1- \frac{1}{d^4}  \Tr (U^{R_2}U^{R_2\dagger})^2 \right).
 \label{op_e}
\end{equation}
As we are interested in nontrial eigenvalues $\lambda_{i\neq 0}$, lets express $\tilde{\ml} = \ml - \ketbra{I}$ and the spectrum of $\tilde{\ml}$ represents eigenvalues of $\ml$ excluding $\lambda_0 =1$, 
\begin{equation}
 \norm{\tilde{\ml}}^2 = \norm{\ml}^2 - 1 = (d^2 -1) (1 - E(U)).  
\end{equation}
Consider the Schur decomposition of $\tilde{\ml} = V_1TV_2$, where $V_1$ and $V_2$ are unitary operators and $T$ is an upper triangular matrix with diagonal entries $\lambda_i$. The Hilbert-Schmidt norm is unitarily invariant which implies that $\norm{\ml} = \norm{T}$, and since $T$ is upper triangular matrix, it follows that 
\begin{equation}
    \sum_{i=1}^{d^2-1} \abs{\lambda_i}^2 \leq (d^2 -1) (1 - E(U)). 
    \label{eq:ineq}
\end{equation}

Interesting results follow from the inequality~(\ref{eq:ineq}). $\abs{\lambda_{d^2-1}}^2 (d^2-1) \leq \sum_{i=1}^{d^2-1} \abs{\lambda_i}^2 \leq (d^2 -1) (1 - E(U))$ and we see that 
\begin{equation}
 \abs{\lambda_{d^2-1}} \leq \sqrt{1-E(U)}. 
\end{equation}
This captures that whenever $E(U) \neq 0$, at least there exists one mode in which $\abs{\lambda_{d^2-1}} < 1$. % hence guarantees non-ergodic nature. Thus $E(U)$ measures the amount of deviation from ergodicity. 

The upper bound for $\abs{\lambda_1}$ can be obtained similarly as 
$\abs{\lambda_1}^2 \leq \sum_{i=1}^{d^2-1} \abs{\lambda_i}^2 \leq (d^2 -1) (1 - E(U))$, 
\begin{equation}
 \abs{\lambda_1} \leq \sqrt{(d^2 -1) (1 - E(U))}. 
 \label{eq:up_bnd}
\end{equation}
From this, it follows that whenever 
\begin{equation}
 E(U) > E_* = \frac{d^2-2}{d^2-1},  
 \label{eu_bound}
\end{equation}
guarantees $\abs{\lambda_1} <1$, which imply mixing. Thus, this provides a sufficient condition for the quantum channel to be mixed. A similar condition for discrete Hamiltonian dynamics holds for mixing in the many-body spacetime dual circuits in terms of entangling power~\cite{ASA_2021}. 

For general $\sigma$, the spectrum is bounded as 
\begin{equation}
    \sum_{i=1}^{d^2-1} \abs{\lambda_i}^2 \leq \Tr\left[ (U^{R_2} (I \otimes \sigma ) U^{R_2\dagger})^2 \right] -1,  
\end{equation}
which provides the sufficiency condition for mixing $\abs{\lambda_1} <1$ if 
\begin{equation}
    \Tr\left[ (U^{R_2} (I \otimes \sigma ) U^{R_2\dagger})^2 \right] < 2.
\end{equation}
If $U$ is {\it dual-unitary}, $U^{R_2}U^{R_2\dagger} = I$, then the purity of the state $\sigma$ determines mixing nature, $\abs{\lambda_1}<1$ if 
\begin{equation}
    \Tr (\sigma^2) < \frac{2}{d}. 
    \label{eq:bound}
\end{equation}
For $d=2$, any mixed state guarantees the mixing, and for general $d$, the bound (\ref{eq:bound})  provides the sufficiency condition for mixing. 

\subsection{Local unitary stability of mixing condition \label{sec:stab}}
Local unitary operators act as a disorder. Two unitary operators $U$ and $U^\prime$ are said to be local unitarily (LU) equivalent if $U^\prime$ can be written as 
\begin{equation}
    U^\prime = (u_1 \otimes u_2) U (u_3 \otimes u_4) 
    \label{eq:LU}
\end{equation}
where $u_i \in \mathcal{H}^d$ are any unitary operators. The spectrum of the quantum chennel $\e (\rho)$ with $U$ and $U^\prime$ with $\e^\prime (\rho) = \Tr_2 [U^\prime (\rho \otimes \sigma ) U^{\prime \dagger}]$ are not identical. The linear operator $\ml^\prime$ for the channel $\e^\prime$ is 
\begin{equation}
\begin{split} 
&\ml^\prime = \left[ U^{\prime R_2} (I \otimes \sigma ) U^{\prime R_2 \dagger} \right]^{R_2} \\
%& = \left[ (u_1 \otimes u_3^T) U^{R_2} (u_2^T \otimes u_4) \left(I \otimes \sigma \right) (u_2^T \otimes u_4)^\dagger U^{R_2\dagger} (u_1\otimes u_3^T)^\dagger  \right]^{R_2} \\
&= \left[ (u_1 \otimes u_3^T) \left[U^{R_2} (I \otimes u_4 \sigma u_4^\dagger ) U^{R_2\dagger}\right] (u_1^\dagger \otimes u_3^*) \right]^{R_2} \\
&= (u_1 \otimes u_1^*) \left[ U^{R_2} (I \otimes  \sigma^\prime ) U^{R_2\dagger}  \right] (u_3 \otimes u_3^*)
\end{split} 
\end{equation} 
where $ \sigma^\prime = u_4 \sigma u_4^\dagger$. Here we have used the identity 
\begin{equation}
    \left[(u_1 \otimes u_2 ) A (u_3 \otimes u_4)\right]^{R_2} = (u_1 \otimes u_3^T) A^{R_2} (u_2^T \otimes u_4 ). 
\end{equation}
 However, it can be seen that $\norm{\ml}^2 = \norm{\ml^\prime}^2$ and hence the sufficient condition for mixing remains invariant under any local unitary operations of the form Eq.~(\ref{eq:LU}). For $\sigma= I/d$, the eigenvalue $\abs{\lambda_1}$ is upper bounded by the operator entanglement $E(U)$ and $E(U)$ is LU equivalent quantity. Even for general $\sigma$ it can be seen the sufficient condition is invariant with local unitary operators. The bound provides stable mixing criteria that are invariant under any local unitary operators. 

\subsection{Block diagonal construction of quantum channels \label{sec:diagon}}
We further study the ergodic hierarchy of channels using block diagonal unitary operators. Define block diagonal unitary operator $U_D$ as $    U_D = \bigoplus_{i=1}^d u_i $ in which $u_i$ are $d \times d$ unitary operators. For $\sigma = I/d$ the linear operator $\ml [U_D]$ ,
\begin{equation}
    \begin{split}
        \ml [U_D] & = \frac{1}{d} [U_D^{R_2}U_D^{R_2\dagger}]^{R_2}. \\ 
        %& = \frac{1}{d} \bigoplus_{i,j=1}^d [u_i^{R_2}u_j^{R_2\dagger}]^{R_2}.  
    \end{split}
\end{equation}
The $U_D^{R_2}$ of block diagonal matrix contains $d$ rows and the inner product of these rows forms the $d^2$ elements of $U_D^{R_2} U_D^{R_2\dagger}$, and hence $\ml [U_D]$ is diagonal with diagonal entries 
\begin{equation}
 \lambda_{ij} = \frac{1}{d} \Tr (u_iu_j^\dagger).
\end{equation} 
There are $d$ eigenvalues $\lambda_{ii}=1$, and hence the channel is non-ergodic. If we choose an orthonormal unitary operator set $u_j$, then there are $d$ eigenvalue $\lambda_{ii} =1$, and the remaining eigenvalues vanish. This is an interesting case as the channel is non-ergodic as a whole but there are $d$ eigenmodes which is highly mixing. If $u_j = u$, all the eigenvalues $\lambda_{ij} =1$ and an example of integrable channel. 

\subsection{Bound on the convergence \label{sec:conv}}
The convergence of any state $\rho$ to the fixed state $\rho_*$ under the action of map $\e$ 
can be bounded as~\cite{terhal2000problem}
\begin{equation}
    \norm{\e^n (\rho) -\rho_*}_1 \leq C_d \text{poly}(n) \abs{\lambda_1}^n, 
\end{equation}
where $C_d$ is some constant depending on the local dimensions $d$. Due to the upper bound on  $\abs{\lambda_1}$ in Eq.~(\ref{eq:up_bnd}), we can write 
\begin{equation}
    \begin{split}
         \norm{\e^n(\rho) -\rho_*}_1 &\leq C_d \text{poly}(n) \left((d^2-1)(1-E(U))\right)^\frac{n}{2} \\
         & \leq C(d,n) \left(1-E(U)\right)^\frac{n}{2}. 
    \end{split}
    \label{eq:convbound}
\end{equation}
The operator entanglement $E(U)$ thus provides an upper bound to the convergence of $\rho$ to $\rho_*$ for a fixed iteration $n$.

\subsection{Two-qubit unitary operators and single qubit quantum channels \label{sec:qubits}}

\begin{figure}
    \centering
    \includegraphics[scale=0.8]{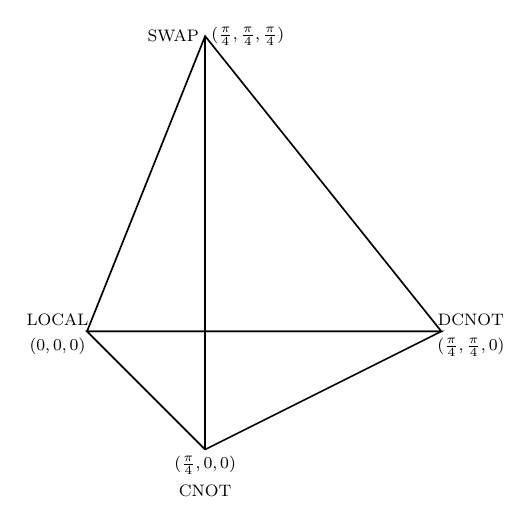}
    \caption{Set of all two-qubit unitary LU equivalent unitary operators represented in the tetrahedron called Weyl chamber. Here the coordinates represent the parameters of unitary operators in Eq.~(\ref{eq:cartan}).}
    \label{fig:weyl}
\end{figure}

For simplest scenerio of $d=2$, any two qubit unitary operator $U$ is unitarily equivalent to $U_{nl}$~\cite{Zhang2003}, 
\begin{equation}
\begin{split}
U_{nl} & = \exp\left[-i (x\,\sigma_x \otimes \sigma_x + y\, \sigma_y \otimes \sigma_y + z\, \sigma_z \otimes \sigma_z)\right] \\
 &= \left( \begin{array}{cccc}
    e^{-iz}c_-       & 0 & 0  & -ie^{-iz}s_- \\
    0       & e^{iz}c_+ & -ie^{iz}s_+ & 0 \\
    0       & -ie^{iz}s_+ & e^{iz}c_+ & 0 \\
    -ie^{-iz}s_-       & 0 & 0  & e^{-iz}c_-  
\end{array} \right), 
\end{split}
\label{eq:cartan}
\end{equation}  
where $c_\pm = \cos(x \pm y), s_\pm = \sin(x \pm y)$ and  
$0 \leq |z| \leq y \leq x \leq \frac{\pi}{4}$. For such $U_{nl}$ with $\sigma = I/d$, $\ml$ can be explicitly diagonalized as 
\begin{equation}
    \ml = \sum_{i=0}^{3} \tilde{\lambda}_i \ketbra{\phi_i}, 
    \label{eq:qubcha}
\end{equation}
where 
 \begin{equation}
     \begin{split}
         \ket{\phi_0} = \frac{\ket{00}+\ket{11}}{\sqrt{2}} & \quad 
         \ket{\phi_1} = \frac{\ket{00}-\ket{11}}{\sqrt{2}} \\
         \ket{\phi_2} = \frac{\ket{01}+\ket{10}}{\sqrt{2}} & \quad 
         \ket{\phi_3} = \frac{\ket{01}-\ket{10}}{\sqrt{2}}
     \end{split}
 \end{equation}
The eigenvalues with trivial  $\lambda_0 = \tilde{\lambda}_0 = 1$ and unordered eigenvalues are 
\begin{equation}
    \begin{split}
        \tilde{\lambda}_1 &= \frac{1}{2} \left[ \cos(2 (x+y)) + \cos(2 (x-y))  \right] \\
        \tilde{\lambda}_2 &= \cos(2y) \cos(2z), \quad  \tilde{\lambda}_3 = \cos(2x) \cos(2z). 
    \end{split}
\end{equation}
The variables $(x,y,z)$ in Eq.~(\ref{eq:cartan}) uniquely represent a set of LU equivalent unitary operators~\cite{Zhang2003}. It forms a tetrahedron called the Weyl chamber, as shown in Fig.~(\ref{fig:weyl}). The ergodic properties of qubit quantum  channel (\ref{eq:qubcha}) can be  characterized using these variables as :
\begin{outline}[enumerate]
    \1 {\bf Case 1:} for the line 'LOCAL' to 'CNOT' :: for which $(x,y,z)$ varies from $(0,0,0)$ to $(\pi/4,0,0)$. The ordered eigenvalues can be then $\lambda_1 = \tilde{\lambda}_2 = 1$ and $\lambda_2 = \lambda_3 = \tilde{\lambda}_1 = \tilde{\lambda}_3 = \cos 2x $. The entire line represents a non-ergodic channel. For $x=0$, local, $\lambda_2 = \lambda_3 = 1$ represents integrable channel. 

    \1 {\bf Case 2:} for the line 'CNOT' to 'DCNOT' :: for which $(x,y,z)$ varies from $(\pi/4,0,0)$ to $(\pi/4,\pi/4,0)$. We have $\lambda_1 = \tilde{\lambda_2} = \cos 2y$ and $\lambda_2= \lambda_3 = \tilde{\lambda}_1 = \tilde{\lambda}_3 = 0 $. Except $y=\pi/4$ i.e.,  CNOT gate,   all other unitaries with will have  $\lambda_1 = \cos(2y) < 1$. Single qubit channels with these unitaries forms mixing channels.
    
    \1 {\bf Case 3:}  The qubit channel corresponding to the unitary operators from line from SWAP to DCNOT for which $(x,y,z)$ varies from $(\pi/4,\pi/4,0)$ to $(\pi/4,\pi/4,\pi/4) $, represents completely depolarizing channel with $\lambda_i =0 \quad \forall i\neq 0$. 
    
    \1 {\bf Case 4:} For the line LOCAL to DCNOT and from LOCAL to SWAP all represent mixing vertices.  
    
\end{outline}

\begin{figure}
    \centering
    \includegraphics[width=0.45\textwidth]{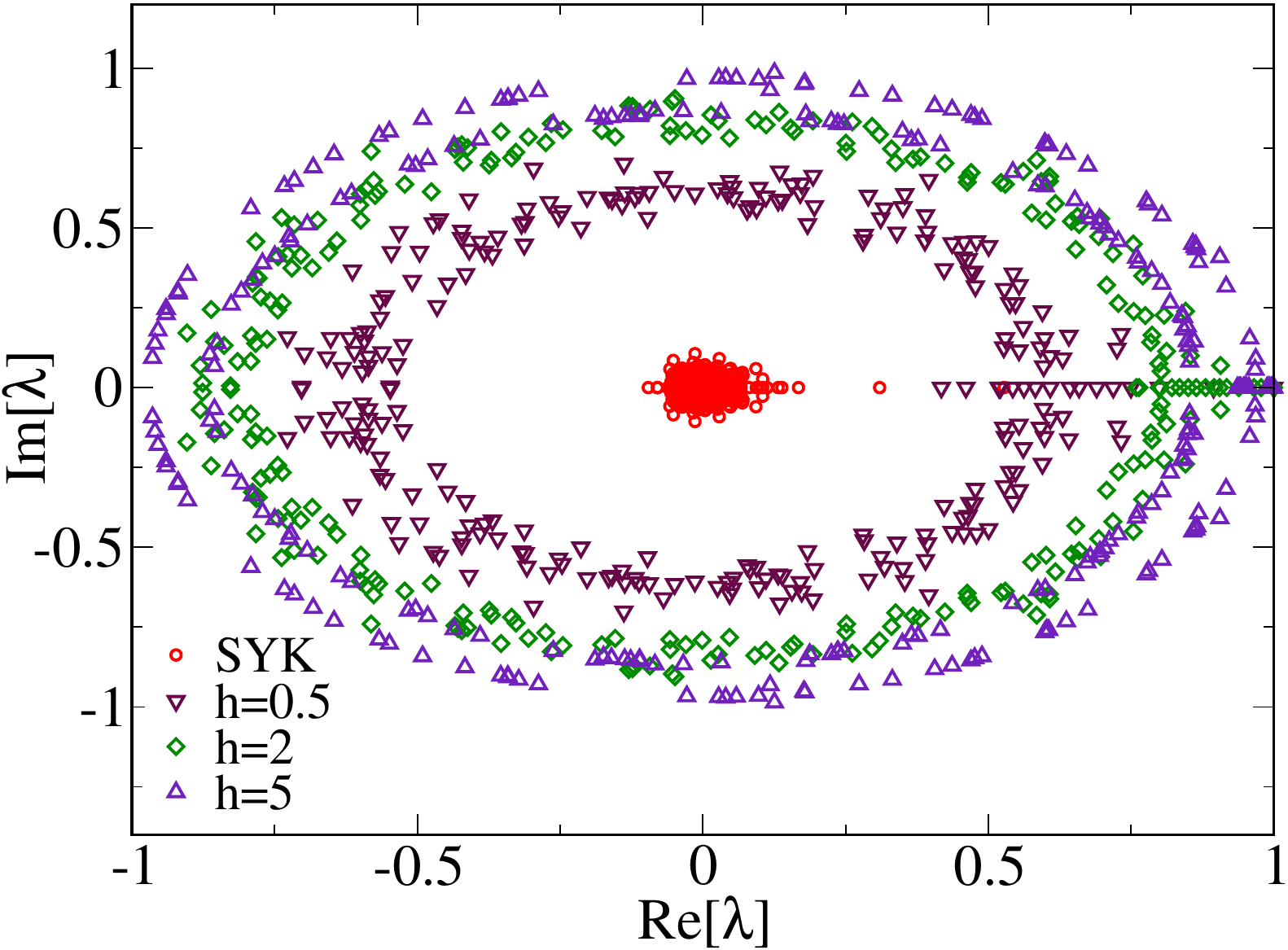}
    \caption{Variation of the real part vs. the imaginary part of the eigenvalue spectrum of the matrix representation of quantum channel $\ml_H$ for the SYK model and the Hamiltonian $H_{SR}$. }
    \label{fig2}
\end{figure}

\begin{figure}
    \centering
    \includegraphics[width=0.45\textwidth]{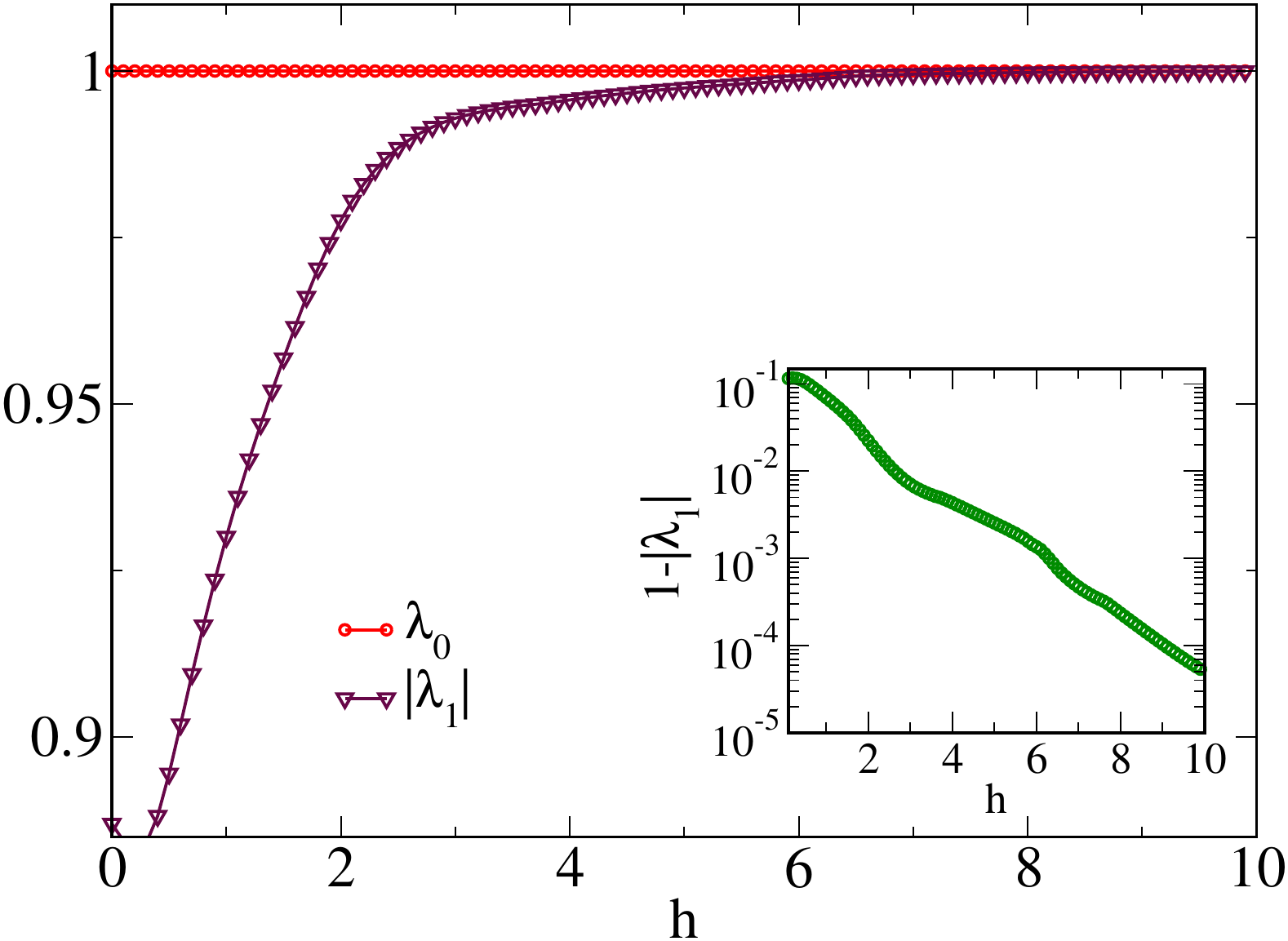}
    \caption{Variation of the 1st two eigenvalues of $\ml_{H}$ $\lambda_0$ and $|\lambda_1|$ (after arranging in descending order in magnitudes) with $h$ for the Hamiltonian $H_{SR}$. Inset shows the approach of $|\lambda_1|$ to 1 with increasing $h$ in semi-log scale.}
    \label{fig3}
\end{figure}

\begin{figure}
    \centering
    \includegraphics[width=0.45\textwidth]{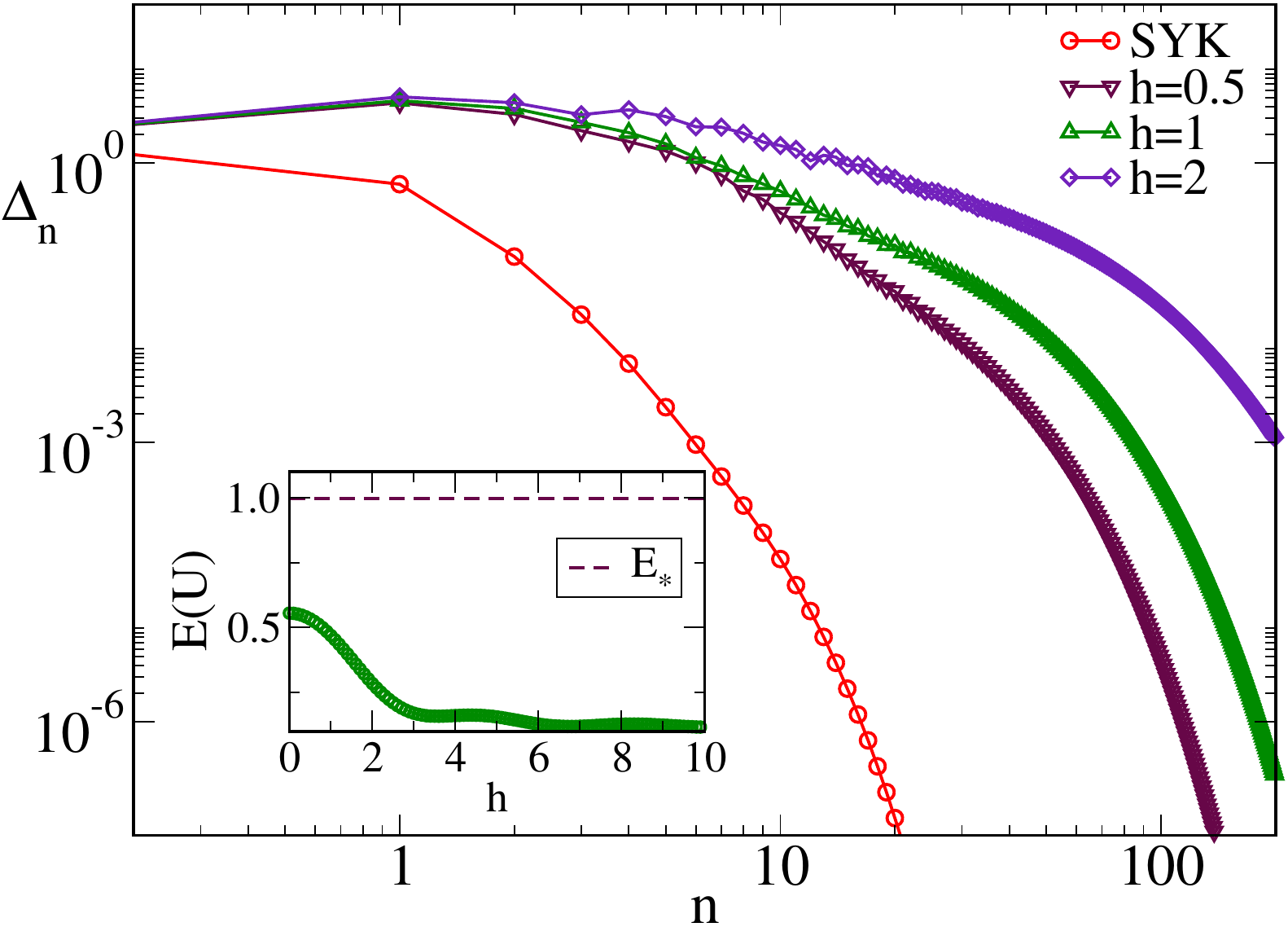}
    \caption{Variation of $\Delta_n=||\ml_H^n \rho-\rho_*||_1$ with $n$ for the SYK model and the Hamiltonian $H_{SR}$ and for the initial Neel state. Inset shows the variation of the operator entanglement with $h$ for the Hamiltonian $H_{SR}$. }
    \label{fig4}
\end{figure}

\begin{figure}
    \centering
    \includegraphics[width=0.45\textwidth]{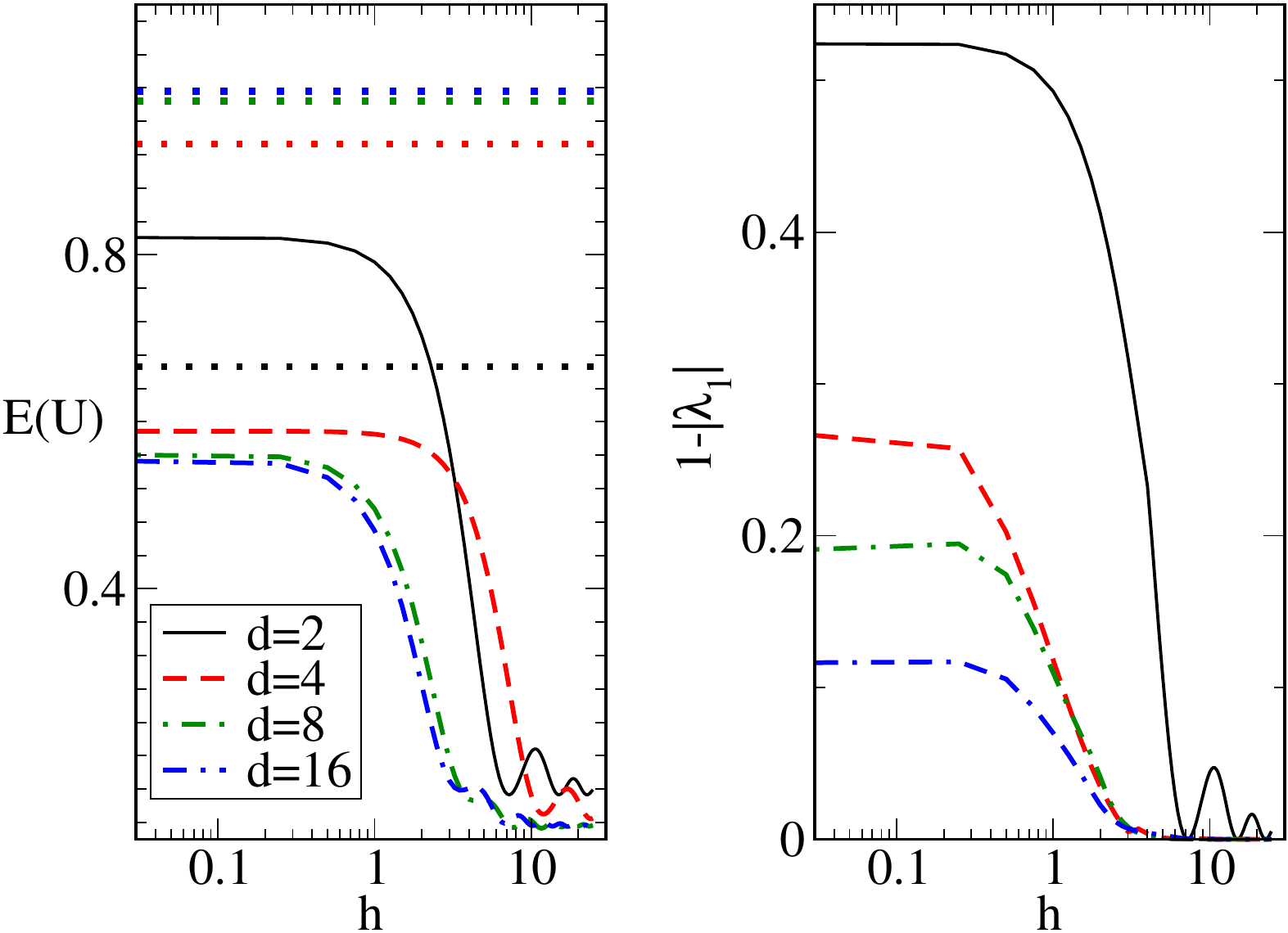}
    \caption{  (Left panel) Variation of operator entanglement $E(U)$ with $h$ and for the Hamiltonian $H_{SR}$ for different values of $d$. The dotted lines correspond to $E_*$. (Right panel)  Variation of the spectral gap  $1-|\lambda_1|$ with $h$ for the Hamiltonian $H_{SR}$ and for different values of $d$.}
    \label{fig4_1}
\end{figure}

%\section{Diagonal unitary operators} 
 
\section{Ergodic and mixing channels with many-body Hamiltonians \label{sec:manybody}}
%\subsection{Exact dynamics}
 In the previous section, we study a two-qubit system, i.e., the Hilbert space dimension of the system  $d=2$ (the Hilbert space dimension of the system+bath $d^2=4$). In this section, we focus on the many-body quantum systems where $d\geq2$.
We study two types of many-body Hamiltonian of interacting fermions on a one-dimensional lattice of length $L$. 
The first one is  an experimentally realizable~\cite{schreiber2015observation,abanin2019colloquium} short-range (SR) model consisting of nearest-neighbor hopping and nearest-neighbor interaction in the presence of quasi-periodic potential. The other one is a model with random all-to-all interactions, famously known as the SYK model.
The Hamiltonians for these two models are given by:
   \begin{eqnarray}
 {H_{SR}}=-\sum _{i=1}^{L-1}(\hat{c}^{\dag}_i\hat{c}^{}_{i+1}+\text{H.c.})+ V\sum_{i=1}^{L-1}\hat{n}_i\hat{n}_{i+1} \nonumber\\
 +h\sum_{i=1}^{L}\cos(2\pi\alpha i) \hat{n}_i, 
 \label{hamiltonian1}
\end{eqnarray}
and 
 \begin{eqnarray}
 {H_{SYK}}=\sum _{i,j,k,l}^{L}J_{ij;kl}\hat{c}^{\dag}_i\hat{c}^{\dag}_j\hat{c}^{}_{k} \hat{c}^{}_{l}, 
 \label{hamiltonian2}
\end{eqnarray}
Where, $\hat{c}^{\dag}_i$ ($\hat{c}_{i}$) is the fermionic creation (annihilation) operator at site $i$, $\hat{n}_i=\hat{c}^{\dag}_i\hat{c}_{i}$ is the number operator.  $\hat{c}^{\dag}_i$  and $\hat{c}_{i}$ obey following anti-commutation 
relations i.e. $\{\hat{c}_i,\hat{c}^{\dag}_j\}=\delta_{ij}$, $\{\hat{c}_i,\hat{c}_j\}=0$, and $\{\hat{c}^{\dag}_i,\hat{c}^{\dag}_j\}=0$. In the case of the Hamiltonian $H_{SR}$, $\alpha$ is an irrational number which we choose to be $\alpha=(\sqrt{5}-1)/2$. We also set the interaction parameter $V=1$ for all our calculations.  
In the absence of the interaction, i.e., $V=0$, $H_{SR}$ undergoes a delocalization-localization transition as one tunes the strength of the quasi-periodic potential, and $h=2$ turns out to be the transition point~\cite{aubry1980analyticity}. 
In the presence of interaction, it is believed that within an isolated quantum dynamics framework, changing the strength of the quasi-periodic potential $h$, one can drive the system from ergodic to many-body localized (MBL) phase~\cite{iyer.2013}, though, in recent days, the stability of MBL transition in the thermodynamic limit has been questioned~\cite{vidmar.2020}. 

On the other, the SYK Hamiltonian $H_{SYK}$ in the large $L$ limit is exactly solvable~\cite{sachdev.93,parcollet.97}.
The study of this model has uncovered a remarkable
direct connection to quantum gravity with a
black hole in $\text{AdS}_2$ ~\cite{subir.10,subir.15,maldacena.16}. As in the case of a black hole,
the SYK model gives
rise to thermalization and many-body quantum chaos with
a Lyapunov exponent that saturates
the quantum limit~\cite{maldacena2016bound}.
The random four-fermion coupling
$J_{ij;kl}$ drawn from a Gaussian distribution with zero mean and variance $|J_{ij;kl}|^2
 = J^2/L^3$; $J_{ij;kl}$ are properly antisymmetrized,
i.e., $J_{ij;kl}=-J_{ji;kl}=-J_{ij;lk}$ and $J_{ij;kl}=J^{*}_{kl;ij}$ ~\cite{sachdev.2015}. For all our calculations, we have taken $J=1$, and all the data for the SYK model are generated after doing an average of 100 random realizations.
 It is important to point out that though both of these systems are fermionic in nature, using Jordan-Winger (JW) transformation, they can be mapped into models of spin-1/2  particles (qubit),  and the issues related to Fermion to qubit mappings~\cite{banuls2007entanglement,gigena2015entanglement,vidal2021quantum} will not impact our findings.  Interestingly, even though JW transformation is non-local, the Hamiltonian $H_{SR}$ becomes the spin-1/2 XXZ model in the presence of the incommensurate magnetic field, which involves spin-exchange term only between nearest sites. On the other hand, the SYK model has all-to-all interactions. Recently, both the SYK model and the $H_{SR}$ model have been realized in quantum computing platforms using proper Qubitization and Trotterization ~\cite{qcxxz,qcsyk1,qcsyk2}. 

We consider the first $L/2$ sites as the system and the rest of the $L/2$ sites as the bath. Hence, the Hilbert space dimension of the sub-system and the bath are $d=2^{L/2}$.
To understand the ergodic hierarchy of these systems according to the formalism described in Sec.\ref{formalism}, we need to investigate the matrix $\ml$ (see Eq.~\ref{eq:L}). Here, for all the calculations, we consider an infinite temperature bath; hence, we replace $\sigma$ in 
Eq.~\ref{eq:L} by $I_d/d$ (where $I_d$ is a $d$ dimensional identity matrix) and the Unitary operator 
$U=e^{iH}=U_H$ (where $H=H_{SR}$ or $H=H_{SYK}$).  Eq.~\ref{eq:L} reads as, 
\begin{equation}
 \ml_H = \frac{1}{d}\left[ U_H^{R_2}  U_H^{R_2 \dagger}\right]^{R_2}.
 \label{eq:L_I}
\end{equation}

First, we investigate the eigenvalue spectrum of the matrix 
$\ml_H$. Figure~\ref{fig2} shows the eigenvalue spectrum of $\ml_H$ in the complex plane, where we plot the real part vs. the imaginary part of the eigenvalues $\lambda_i$. We find that for the Hamiltonian $H_{SR}$, as we increase $h$, the effective distance from the center of most of the data points increases, indicating the absolute value of a large fraction of eigenvalues $|\lambda_i|$ approaches toward $1$. On the other hand, for the SYK model, most of the data points are in the neighborhood close to the center. 
Intuitively, one would also expect similar results, given the Hamiltonian $H_{SYK}$ 
is the maximally chaotic model that saturates the quantum limit of the Lyapunov exponent~\cite{maldacena2016bound}. It is also important to point out that 
while we have related the inverse of mixing time (mixing speed) with the spectral gap $1-|\lambda_1|$ in the rest of the manuscript, 
one could also use some other definition, such as $1-\frac{1}{d^2-1}\sum_{i=1}^{d^2-1}|\lambda_i|$; 
the average distance of the absolute values of all the $\lambda_i$s from $\lambda_0$. Given the maximally chaotic nature of the SYK model, one would expect an extremely fast mixing for this model; hence, the indicator $1-\frac{1}{d^2-1}\sum_{i=1}^{d^2-1}|\lambda_i|$ (or even $1-|\lambda_1|$) should have a large value. All the eigenvalues are concentrated near the origin, indicating the mixing speed measure will show a much larger value for the SYK model compared to $H_{SR}$ models.

Moreover, we arrange the spectrum in descending order of magnitude (as mentioned in Eq.~\eqref{order}), plot $\lambda_0$ (which is real) and $|\lambda_1|$ as a function of $h$ in Fig.~\ref{fig3}. While $\lambda_0=1$, $|\lambda_1|$ approaches towards $1$ with increasing $h$. However, the inset of Fig.~\ref{fig3} ensures that for a finite value of $h$, $|\lambda_1|$ does not reach $1$, though $1-|\lambda_1|$ approaches to zero with increasing $h$. It is important to emphasize that in the $h\to\infty$ limit, $U_H$ will be completely diagonal. Hence, it is a non-ergodic channel~(see Sec. (\ref{sec:diagon}). 
$|\lambda_1|$ for the SYK model for $L=8$ turns out to be $0.5275$, which is much smaller than $|\lambda_1|$ obtained for $H_{SR}$.  Intuitively, one expects that given the SYK model is maximally chaotic. This confirms that the quantum channel representations of the Hamiltonian $H_{SR}$ and $H_{SYK}$ show mixing. 

To strengthen the claim, 
 we also plot $\Delta_n=||\ml_H^n \rho-\rho_*||_1$ vs $n$. For a mixing channel $\Delta_n\to 0$ as $n\to \infty$, $\forall \rho$ and where $\ml_H \rho_*=\rho_*$. The results shown in Fig.~\ref{fig4} are for Neel state $|\psi\rangle=\Pi_{i=1}^{L/2}\hat{c}^{\dag}_{2i-1} |0\rangle$ ~\cite{alba2018entanglement,modak.2023,modak2020entanglement}; however, we check our results for a few other states as well. We find that, indeed, $\Delta_n$ approaches zero as $n$ increases. 
 As expected, with increasing $h$ in $H_{SR}$, the fall-off of $\Delta_n$ becomes slower. On the other hand, the decay of  $\Delta_n$ is extremely fast for the SYK model. Moreover, we also study the variation of the operator entanglement $E(U)$ (see Eq.~\eqref{op_e}) with $h$ for the Hamiltonian $H_{SR}$ in the inset of Fig.~\ref{fig4} and find that it decreases with increasing $h$, and always remains below $E_*$ (see Eq.~\eqref{eu_bound}). On the other hand,  for the SYK model (for $L=8$), $E(U)\simeq 0.9925$ is very close to $E_*$. 

Next, we also investigate the effect of the Hilbert space dimension of the system $d$ in the ergodic hierarchy. In Fig.~\ref{fig4_1} (right panel), we show the variation of the spectral gap $1-|\lambda_1|$ (a measure of the mixing speed) with $h$ of the short-range Hamiltonian $H_{SR}$ for different values of $d$. The spectral gap appears to decrease quickly with increasing $d$. It says that as $d$ increases, the inverse mixing time (mixing speed) decreases for the short-range model. Similar features can also be observed in the variation of the operator entanglement vs. $h$ plot in  Fig.~\ref{fig4_1} (left panel). Interestingly, we find that only for $d=2$ and for a certain parameter range ($h\in[0,3]$), $E(U)$ satisfies the sufficient condition of the mixing, i.e., $E(U) > E_*$. For $d>2$, $E(U)$ always appears to be less than $E_*$. Intuitively, these results can be understood in the following way. In the case of the $H_{SR}$, the coupling term $H_{sb}$ (according to schematic Fig.~\ref{schematic}) between the system and the bath involves only two adjacent sites, does not depend on the system size $L$ or the Hilbert space dimension of the system $d$. Hence, one would expect that if $L$ (or $d$) increases, the bath will take more time to thermalize the system; hence, the inverse mixing time is expected to go down with increasing $d$. On the contrary, for the all-to-all SYK model, this is not the case; hence, there is not much of a dependence of the mixing time on $d$.

\begin{figure}
    \centering
    \includegraphics[width=0.45\textwidth]{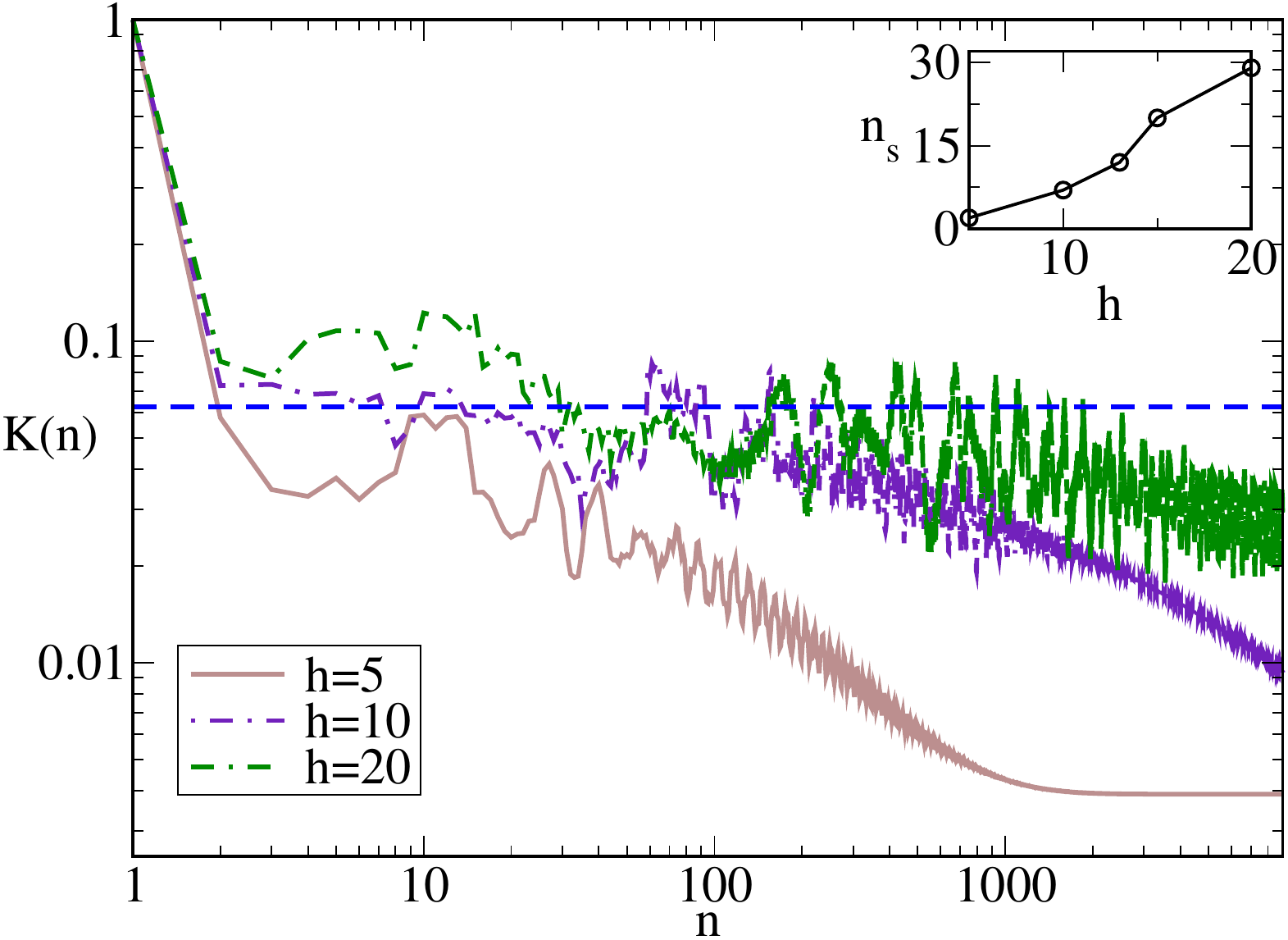}
    \caption{ Variation of generalized SFF $K(n)$ with $n$ for  the Hamiltonian $H_{SR}$. Inset shows the variation of the scrambling time $n_s$ (iteration) with $h$. The dashed horizontal line corresponds to $1/d$.}
    \label{fig5}
\end{figure}

 % While the channel for which  $\eta(\e) <1$, is referred to as the scrambling channels, but $\eta(\e)$ is not an easily computable quantity. The question arises: can one identify some alternative way (preferable from the eigenvalue spectrum of $\ml_H$ \eqref{eq:L_I}, whether a channel is Scrambling or not?  A state- and observable-independent characteristic of
 Further, we also make a quantitative connection between thermalization (measured in terms of scrambling~\cite{PhysRevLett.132.040402}) and the ergodic hierarchy using the spectral form factor (SFF) ~\cite{haake1991quantum,muller2004semiclassical,muller2004semiclassical,vikram2023dynamical}, a statistical measure of quantum chaos with that of the ergodic hierarchy. The generalized SFF for a quantum channel is defined as~\cite{PhysRevLett.132.040402}, 
\begin{equation}
    K := \frac{1}{d^2} \sum_i \abs{\Tr A_i }^2
\end{equation}
where $A_i$ are the Kraus operator of the dynamical equation (see Eq.~(\ref{eq:kraus})). By using a simple identity $\abs{X}^2 = \Tr(X\otimes X^*)$, we can write generalized SFF as   
 \begin{equation}
     K(n)=\frac{1}{d^2}\Tr\ml^n_H,
 \end{equation}
 where $d$ stands for the Hilbert space dimension of the system. It has been argued that the necessary condition for scrambling in many-body systems is  $K(n)<1/d$ ~\cite{PhysRevLett.132.040402}.   In terms of the eigenvalues of $\ml$, the generalized SFF  
 \begin{equation}
     K(n) = \frac{1}{d^2} (1 + \sum_{i=1}^{d^2-1} \lambda_i^n).
 \end{equation}
 It is very interesting to note that in the block diagonal construction of quantum channels constructed in Sec.~(\ref{sec:diagon}), there exists a model for which $d$ eigenvalues equals $\lambda_i = 1$ and remaining eigenvalues vanish giving rise to a non-thermalizing model. From the ergodic theory perspective, the model itself is non-ergodic, and it matches with that of the many-body perspectives.

Given the necessary condition for scrambling has been argued as  $K(n)<1/d$ ~\cite{PhysRevLett.132.040402}, here, we roughly estimate the scrambling time by identifying $n_s=\min(n)$, such as the $K(n_s)\leq 1/d$.  
We have considered $H=H_{SR}$, computed $K(n)$ as a function of $n$. Previously, we have shown in Fig.~\ref{fig2} and Fig.~\ref{fig3}  that for any finite value $h$, 
the channel $\ml_H$ is mixing. However, as we increase the $h$, the spectral gap $1-|\lambda_1|$ (which is also a measure of inverse mixing time) decreases. One would naively expect that mixing time and the scrambling time should be related to each other. If one observes the decrease in mixing time, similar features should also be noticed in the context of scrambling. 
Figure.~\ref{fig5} validates our hypothesis; we show the variation of $K(n)$ vs. $n$ for different values of $h$. As expected, with increasing $h$, the fall-off of $K(n)$ with $n$ is slower. Moreover, we also identify the scrambling time (iteration) $n_s$ and find that, indeed, it increases with increasing $h$.

\bla 
\section{Conclusion and Discussions}

In this work, we study the ergodic hierarchy of quantum channels by proving the necessary condition for the channel to be mixing in terms of operator entanglement of the unitary operator, and we use these tools to study the mixing in quantum many-body systems. We also prove the stability of the sufficiency of the mixing condition under any single particle local unitary operator. We construct various analytical models using block diagonal unitary operators to explore the ergodic hierarchy of quantum channels. As an illustration, by using the canonical two-qubit unitary operator, we construct various quantum channels of a single qubit that show different ergodic behaviors.    
Then, we explore unitary operators constructed from quantum many-body Hamiltonians, especially focused on two models: 1) short-range interacting models in the presence of quasiperiodic potential and 2)  all-to-all interacting SYK model. We find 
both the models show mixing, while SYK, being maximally chaotic, almost satisfies the sufficiency conditions for mixing; the other model fails to do so. Also, studying different spectral measures of quantum channels and the convergence criteria, it was apparent that the mixing rate for the SYK model is much higher compared to the short-ranged model. Intuitively, it fits well with our notion that the SYK is maximally chaotic in the thermodynamic limit. However,  both in the two-qubit case and for the many-body systems, we could not find an example where the channel is ergodic but not mixing. We also study the quantitative relation between mixing and scrambling using spectral form factor. 

In summary, based on our results,  we  make a few conjectures for generic  quantum many-body systems, and they are  
1) In the context of quantum many-body systems, 
%mixing is much more common than other ergodic hierarchies,
we have looked into a broad range of one-dimensional systems that are experimentally realizable in cold-atom experiments. It turns out that all of them show mixing, we could not find an example where the system is ergodic but not mixing.
2) The  mixing speed decreases with increasing $d$ for the short-range model, 
3) As $d$ increases, the sufficient conditions of mixing in terms of the operator entanglement become harder to achieve. 4) The mixing and scrambling times are typically expected to behave similarly for a generic quantum system. However, we must emphasize one can possibly cook up an extremely fine-tuned model where some of our above conjectures may get violated; it will be interesting to search for such models in future studies.

In the last century, the mathematics of ergodic theory has provided many interesting results and opened up many applications~\cite{katok1995introduction}. %Classical ergodic theory is used in classical information theory in various scenarios~\cite{breiman1957individual,chung1961note,parry2004topics,gray2011entropy}.  
It would be intriguing to explore how far one can draw parallels from the world of classical ergodic theory to the world of quantum channels. In particular, one might be able to characterize total ergodicity (all powers of the given transformation are ergodic)~\cite{eisner2015operator} and prove that it sits distinctively in between ergodicity and weak mixing, but on the other hand, the notion of mild mixing~\cite{eisner2015operator} defined in any conceivable way will be equivalent to mixing as it would sit between mixing and weak mixing. It would be interesting to see the relation between the ergodic theory of quantum channels and its information-theoretic capabilities.

\section{acknowledgements}
RM acknowledges the DST-Inspire fellowship by the Department of Science and Technology, Government of India, SERB start-up grant (SRG/2021/002152). SA acknowledges the start-up research grant from
SERB, Department of Science and Technology, Govt. of India (SRG/2022/000467).

\bibliography{channel,ent_local}
\end{document}